\documentclass[runningheads,orivec]{conference/llncs}

\makeatletter
\@twosidefalse
\@mparswitchfalse
\makeatother

\usepackage[a4paper, total={6in, 8in}]{geometry}
\usepackage[T1]{fontenc}
\usepackage{cite}
\usepackage{graphicx}
\usepackage[font=small]{caption}
\usepackage{paralist}
\usepackage{float}
\usepackage{amsfonts}
\usepackage{amssymb}
\usepackage{amsmath}
\usepackage[hidelinks]{hyperref}
\usepackage[capitalise,nameinlink]{cleveref}
\usepackage{subcaption}
\usepackage{algorithm}
\usepackage{algorithmic}

\usepackage{tikz}
\usetikzlibrary{positioning, calc, fit, math, trees}
\usetikzlibrary{decorations.pathreplacing,calligraphy} 
\usepackage{todonotes}
\usepackage{wrapfig}
\usepackage[norefs,nocites,nomsgs]{refcheck}  %

\usepackage{graphicx}

\newcommand{\bigO}{\mathcal{O}}

\newcommand{\Oh}[1]{\bigO\left( #1 \right)}
\newcommand{\Ohn}{\bigO(n)}

\newcommand{\ZP}[1]{Z^{P}_{#1}}

\newcommand{\lc}{\operatorname{lc}}
\newcommand{\rc}{\operatorname{rc}}

\newcommand{\wit}[1]{\widetilde{#1}}

\newcommand{\man}{\mathsf{A}}
\newcommand{\fp}{\mathsf{F}}
\newcommand{\eqf}{G_{=}(\fp)}
\spnewtheorem{observation}[lemma]{Observation}{\bfseries}{\itshape}
\crefname{observation}{Observation}{Observation}

\begin{document}
\title{Combinatorics of Palindromes}
%
%
\author{Michael Itzhaki\inst{1}\orcidID{0009-0009-4783-2537}}
\authorrunning{M. Itzhaki}
%
\institute{Bar-Ilan University
\email{michaelitzhaki@gmail.com}}
\maketitle              
%






\bibliographystyle{conference/splncs04}

\begin{abstract}
We investigate the structure and reconstruction complexity of Manacher arrays. First, we establish a combinatorial lower bound, proving that the number of rooted tandem repeat trees with $n+1$ genes exceeds the number of distinct Manacher arrays of length $n$. Second, we introduce a graph-theoretic framework that associates a graph to each Manacher array, where every proper vertex coloring yields a string consistent with the array. Finally, we analyze a reconstruction algorithm by I et al. (SPIRE 2010), showing that it simultaneously achieves a globally minimal alphabet size, uses at most $\log_2(n{-}1) + 2$ distinct symbols, and can be adapted to produce reconstructions over arbitrary alphabets when possible. Our results also resolve an open problem posed by the original authors. Together, these findings advance the combinatorial understanding of Manacher arrays and open new directions for string reconstruction under structural constraints.

\keywords{Palindrome  \and Combinatorics \and Manacher Array \and Reconstruction}

\end{abstract}

\section{Introduction}\label{s:intro}

Palindromes—strings that read identically forward and backward—arise naturally in fields such as mathematics, computer science, and bioinformatics. Their study sheds light on structural symmetry in strings and has implications for sequence alignment, data compression, and formal language theory~\cite{KMP-77,4567908, Lenz2019,gk:97,f:03,lssnz:05,sr:12}. In computational settings, palindromes are central to problems ranging from error detection in coding theory to modeling self-replicating biological patterns. A classical linear-time algorithm for computing all maximal palindromes in a string is due to Manacher~\cite{man:75}, with further variants addressing the longest palindromic substring~\cite{10.1007/BF01182773,10.1145/270563.571472}, palindromic length~\cite{borozdin_et_al:LIPIcs.CPM.2017.23,10.1007/978-3-319-07566-2_16}, efficient indexing~\cite{10.1007/978-3-319-29516-9_27}, and sublinear-time detection~\cite{charalampopoulos_et_al:LIPIcs.CPM.2022.20}.


{\em Tandem Duplication Trees} provide a combinatorial model of genome evolution through repeated duplications of adjacent segments~\cite{b:99}; see also~\cite{EICHLER2001661,10.1093/molbev/msh115}. We uncover a structural connection between Manacher arrays and rooted tandem duplication trees.

Relations between strings and graphs have been explored before~\cite{10.1007/978-3-540-45138-9_15,GKRR:20}; we extend this programme by attaching a graph to every Manacher array and analyzing its vertex colorings. This perspective yields a reconstruction framework for Manacher arrays based solely on the vertex coloring of the graph.

In this work, we establish a combinatorial link between Manacher arrays and duplication trees, introduce a graph-based reconstruction framework via vertex colorings, and prove that a classical algorithm achieves minimal alphabet size. Our results resolve an open problem and deepen the combinatorial understanding of palindromic structures and Manacher arrays.


\paragraph{Our contributions.}
%
%
%
%


\begin{itemize}
    \item \textbf{Combinatorial Lower Bound.} We prove a surprising combinatorial bound relating Manacher arrays and rooted tandem repeat trees (Theorem~\ref{thm:tandem-trees}): the number of rooted tandem repeat trees with $n+1$ genes is greater than or equal to the number of distinct Manacher arrays of length $2n-1$. The combinatorial problem of counting the number of unique Manacher arrays is motivated by the literature~\cite{tomohiro:10}.
    
    \item \textbf{Graph-Theoretic Reconstruction.} We introduce a novel graph construction for Manacher arrays (Theorem~\ref{thm:graph-reconstruction}), where:
    \begin{itemize}
        \item Every Manacher array can be represented by a graph.
        \item Every proper coloring of the graph yields a unique string corresponding to the given Manacher array.
    \end{itemize}

    \item \textbf{Minimal Reconstruction Algorithm.} We analyze the algorithm of I et al.~\cite{tomohiro:10} that reconstructs a string from its Manacher array and prove the following (Theorem~\ref{thm:minimal-reconstruction}):
    \begin{itemize}
        \item It outputs a string with the globally minimal alphabet size, achieving the given Manacher array.
        \item The number of distinct symbols used is at most $\log(n-1) + 2$, and there exists a tight example for every $n$.
        \item The algorithm can be modified to generate a string with any desired alphabet size (when such a string exists).
    \end{itemize}
    
\end{itemize}


\section{Preliminaries}\label{s:pre}

An ordered sequence of characters is a {\em string}. A string $S$ of length $|S|=n$ is the sequence $S[1]S[2]\dots S[n]$ where $\forall 1\le i\le n,\;S[i]\in \Sigma$. The set $\Sigma$ is called the {\em Alphabet} of $S$. The {\em empty string} with $|S|=0$ is denoted as $\varepsilon$. A {\em substring} of a string is the string $S[i..j]$ where $1 \le i \le j \leq n$ and it is formed from the characters of $S$ starting at index $i$ and ending at index $j$, i.e., $S[i..j] = S[i]S[i+1]\dots S[j]$. If $j < i$, $S[i..j]$ is the empty string $\varepsilon$. A prefix of $S$ is a substring $S[1..i]$, and a suffix of $S$ is a substring $S[i..n]$. We say that a character $c$ {\em occurs} in $S$ if and only if $\exists i$ s.t. $S[i] = c$, and denote $c\in S$.

For an integer $i \in \mathbb{N}$, we denote as $S^i$ the concatenation of a string to itself $i$ times, i.e., $S^1 = S$, and $S^i = S^{i-1} \cdot S$, for any $i\geq 2$. A string $S=p^ip'$ is called {\em periodic}, where $i \ge 2$ and $p'$ is a prefix of $p$. The substring $p$ is called the {\em period} or {\em factor} of $S$. A string with two periods $p,\;q$ has a period of length $\gcd(|p|,|q|)$ (Fine \& Wilf~\cite{fw:65}).

We call $S$ a \textit{palindrome} or \textit{palindromic} if $\forall 1\le i\le n,\; S[i]=S[n-i+1]$. For example, $S=\texttt{level}$, or $S=\texttt{deed}$ are palindromes. However, $S=abcbaa$ is not a palindrome. A substring $P=S[i..j]$ is called a {\em maximal palindrome} of $S$ if $P$ is a palindrome, and $S[i-1..j+1]$ is either undefined or not a palindrome. The center of $P$ is $c_P = \frac{i+j}{2}$ and its radius is $r_P=\lceil\frac{j-i}{2}\rceil$. 


There are $2n-1$ possible centers in $S$ and each has exactly one corresponding maximal palindrome. 
We call a palindrome $p$ of length $\le 1$ a {\em trivial palindrome}. 
An array of all maximal palindromes can be found in linear time, e.g., using the Manacher algorithm~\cite{man:75}. The that retains the length of the maximal palindrome for every given center is known as the {\em Manacher array}, also referred to in the literature as the {\em palindromic structure} of $S$.

\begin{definition}[Manacher array]
Let $S[1 .. n]$ be a string of length $n$. The \emph{Manacher array} of $S$ is the array $\mathsf{A}[1 .. 2n-1]$ defined as follows:

Each position $i$ corresponds to a center between characters of $S$:
\begin{itemize}
    \item If $i$ is odd, $i = 2k-1$, then $\mathsf{A}[i]$ is the maximum integer $r \geq 0$ such that $S[k - r..k + r]$ is a palindrome, and $1 \le k - r \le  k + r \le n$.
    \item If $i$ is even, $i = 2k$, then $\mathsf{A}[i]$ is the maximum integer $r \geq 0$ such that $S[k - r + 1..k + r]$ is a palindrome, and $1 \le k - r + 1 \le  k + r \le n$.
\end{itemize}
\end{definition}

In other words, $\mathsf{A}[2k]$ is the radius of the longest palindrome centered between $S[k]$ and $S[k+1]$, and $\mathsf{A}[2k-1]$ is the radius of the longest palindrome centered at $S[k]$.

\begin{lemma}[Folklore,~\cite{10.1007/978-3-319-07566-2_16}]\label{p:pal-per}
Let $P$ be a palindrome. $p$ is a palindromic suffix of $P$ iff $|P|-|p|$ is a period of $P$.
\end{lemma}

\begin{definition}\label{def:pal-dep}
    Let $S$ be a string. An index $m$ is said to be {\em palindromically dependent} iff there exists a palindromic substring $p=S[i..j]$ such that $c_p<m\le j$.
\end{definition}

The \textit{Zimin word} \( z_n \) is a recursively defined sequence of words, where:
\[
z_1 = w_1 \quad \text{and} \quad z_{n+1} = z_n w_{n+1} z_n,
\]
with each \( w_i \) being a distinct non-empty variable. The \textit{Zimin degree} of a word \( w \) is the largest number \( k \) such that \( w \) can be written as \( z_k \) (i.e., \( w \) is a substitution instance of \( z_k \)). A Zimin word of degree \( k \) has a length of at least \( 2^k - 1 \).

\begin{definition}[Palindromic Zimin Word]
Let $\ZP{1}$ be any arbitrary palindrome, and let $\ZP{k}$ be the pattern of all palindromes that match $\ZP{k-1} P_k \ZP{k-1}$ where $P_k$ is an arbitrary palindrome, and subsequent choices of $P_i$ satisfy $P_i\neq P_k$. A word that matches the structure $\ZP{k}$ is a {\em Palindromic Zimin Word}.
\end{definition}

Since every word that follows the structure $\ZP{k}$ also follows $\ZP{k-1}$, we say that $S$ follows (or matches) $\ZP{k}$ if it follows $\ZP{k}$ and does not follow $\ZP{k+1}$ \footnote{This pattern is also known as the ``ABACABA pattern''.}.


\paragraph{Tandem Duplication Trees.}
Let $T$ be a rooted binary tree and let $v \in T$ be a node.
We write $\operatorname{par}(v)$ for the parent of $v$, $\lc(v)$ for its left child, and $\rc(v)$ for its right child.
If $v$ is the root, then $\operatorname{par}(v)$ is undefined; if $v$ is a leaf, then $\lc(v)$ and $\rc(v)$ are undefined.

We now define the {\em Tandem Duplication Tree} and {\em Rooted Tandem Duplication Tree}, data structures motivated by biological evolutionary constructs.

\begin{definition}[Tandem-Duplication Event]
Let $\Gamma$ be a universe of gene identifiers, and let
$A = \{g_1,\dots,g_m\}$ be an ordered sequence with $g_j \in \Gamma$. \\
Pick any contiguous block of indices:
\[
B=\{i,i+1,\dots,i+\ell-1\},\quad(1\le i \le i+\ell -1 \le m)
\]
Replace \(A[i..i+\ell-1]\) by:
\[
  \{\,\lc(g_i),\dots ,\lc(g_{i+\ell-1}),\,\rc(g_i),\dots ,\rc(g_{i+\ell-1})\,\},
\]
where for every \(g_j\in B\) the symbols \(\lc(g_j)\) and \(\rc(g_j)\)
are \emph{distinct new} genes. Such an operation is referred to as {\em tandem-duplication event} or simply a {\em duplication event}.
A block $B$ of size \(k=\ell\) is called an \emph{\(\ell\)-duplication}.
\end{definition}

Using the latter definition, we define a {\em rooted duplication tree} as a tree that conserves all duplication events on a single gene.
  
        
\begin{definition}[Rooted Duplication Tree]
Let $A_0=\{1\}$ be the original gene, and let $A$ be the final genes array that results from consecutive tandem duplication events on the original array $A_0$.

The history of all events is stored in a rooted binary tree \(T\):
\begin{compactitem}
  \item the root represents the ancestral gene \(1\);
  \item each internal node corresponds to one duplication event, its left (and right) subtree containing all \(\lc(\,\cdot\,)\)
        (resp.\ \(\rc(\,\cdot\,)\)) descendants created by that event;
  \item leaves are the genes present in the final array.
\end{compactitem}
\end{definition}

An example can be found at~\cref{fig:rooted-tandem-repeat}. 
%
%
%
%
%
%

\begin{figure}
    \centering

\begin{tikzpicture}[
  scale=0.85,
  level distance=10mm, 
  every node/.style = {font=\small},
  treenode/.style = {circle, draw, inner sep=1pt},
  leaf/.style = {inner sep=1pt},
  dashedbox/.style={draw=black, dashed, inner sep=5pt, rounded corners}
]
\node[treenode]{r}
  child[xshift=-10mm] {node[treenode]{a}
    child[xshift=-10mm] {node[treenode]{b}
      child[xshift=-10mm] {node[treenode](d) {d}
        child[xshift=-10mm, yshift=3mm] {node[treenode](h) {h}
          child[yshift=-3mm] {node[leaf]{1}}
          child[yshift=2mm, xshift=2mm] {node[treenode]{q} 
            child[yshift=5mm, xshift=3mm] {node[leaf]{3}}
            child[yshift=5mm, xshift=-3mm] {node[leaf]{4}}
          }
        }
        child[xshift=5mm, yshift=-5mm] {node[treenode]{n}
            child[yshift=5mm, xshift=3mm] {node[leaf]{6}}
            child[yshift=5mm, xshift=-3mm] {node[leaf]{7}}
        }
      }
      child[xshift=-2mm] {node[treenode](f){f}
        child[xshift=-25mm,yshift=3mm] {node[treenode](i){i}
            child[xshift=-3mm, yshift=-3mm] {node[leaf]{2}}
            child[xshift=3mm, yshift=-3mm] {node[leaf]{5}}
        }
      child[xshift=-2mm] {node[treenode](j){j}
        child[xshift=2mm] {node[leaf]{8}}
        child[xshift=1mm] {node[leaf]{10}}
      }
      }
    }
    child[yshift=-20mm, xshift=-5mm] {node[treenode](k){k}
        child[xshift=-1mm] {node[leaf]{9}}
        child[xshift=-2mm] {node[leaf]{11}}
    }
  }
  child[xshift=+10mm] {node[treenode]{c}
    child[yshift=-20mm, xshift=-5mm] {node[treenode]{o}
        child[xshift=3mm] {node[leaf]{12}}
        child {node[leaf]{13}}
    }
    child[xshift=+10mm] {node[treenode]{g}
        child[xshift=+0mm] {node[treenode](l){l}
            child[xshift=+0mm] {node[treenode]{p}
                child[xshift=+5mm] {node[leaf]{14}}
                child {node[leaf]{15}}
            }
            child[yshift=-10mm, xshift=14mm] {node[leaf]{17}}
        }
        child[xshift=+10mm] {node[treenode](m){m}
            child[yshift=-10mm, xshift=-9mm] {node[leaf]{16}}
            child[yshift=-10mm] {node[leaf]{18}}
        }
    }
  };

  \node[dashedbox, fit=(l)(m)] {};
  \node[dashedbox, fit=(j)(k)] {};
  \node[dashedbox, fit=(d)(f)] {};
  \node[dashedbox, fit=(h)(i)] {};

    \end{tikzpicture}
    \caption{Rooted tandem duplication tree. Duplication events with more than one gene are marked with a dashed rectangle. The leaves are labeled and positioned in order.}
    \label{fig:rooted-tandem-repeat}
\end{figure}
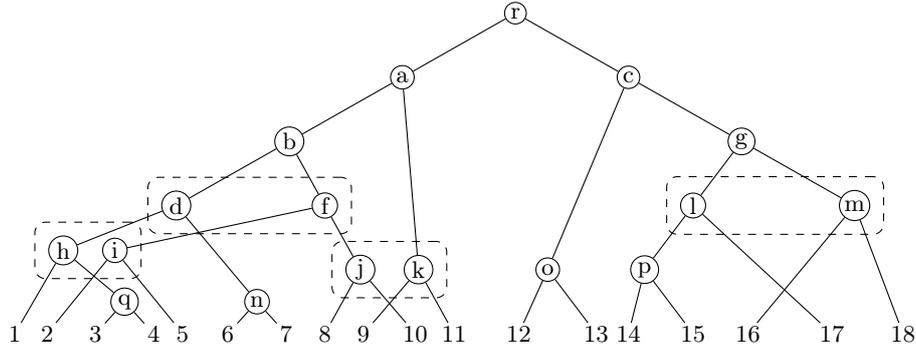

\begin{lemma}[Rooted Duplication Tree's count~\cite{10.1093/molbev/msh115}]
    Let $r_n$ be the number of distinct (non-isomorphic) rooted duplication trees with $n$ leaves. The value $r_n$ follows the recurrence formula:
    \[
    r_n=\begin{cases}
    1 & n=1,2 \\
    \sum_{k=1}^{\lfloor(n+1)/3\rfloor}{(-1)^{k+1}{\binom{n+1-2k}{k}}r_{n-k}} & n\ge 3
    \end{cases}
    \]
\end{lemma}

Using Stirling's approximation~\cite{knuth1994concrete}, it can be shown that $r_n=o(6.75^n)$.

\paragraph{Graph Theory.}  
We use standard graph-theoretic terminology. A graph \(G = (V, E)\) consists of a finite set of vertices \(V\) and a set of edges \(E \subseteq V \times V\), where all edges are undirected and loop-free. A \emph{proper vertex coloring} is a function \(\psi: V \to \Sigma\) such that \(\psi(u) \ne \psi(v)\) for every edge \((u,v) \in E\); the set \(\Sigma\) is referred to as the set of colors. The minimal number of colors required for a proper coloring of \(G\) is its chromatic number, denoted \(\chi(G)\).

\section{Combinatorial Complexity of the Manacher Array}\label{sec:ent}

In this section, we discuss the combinatorial complexity of the Manacher array, i.e., we attempt to find the number $\rho_n$ of distinct Manacher arrays generated from strings of length $n$.

The number of distinct strings of length $n$ with alphabet $\Sigma$ is $|\Sigma|^n$. However, the number of distinct Manacher arrays corresponding to these strings is smaller. We show an upper and lower bound for $\rho_n$. The lower bound is $\Omega\left(3^n\right)$, while the upper bound is $\Oh{r_{n+1}}$, the number of rooted duplication trees with $n+1$ leaves.

\begin{lemma}
Let \(\rho_n\) denote the number of distinct Manacher arrays corresponding to strings of length \(n\). Then \(\rho_n = \Omega(3^n)\).
\end{lemma}
The Lemma follows from Lemma 2 in~\cite{tomohiro:10}, where they prove that a string over a ternary alphabet can be reconstructed from its Manacher array, up to a permutation.

We now proceed to prove the upper bound with rooted tandem repeat trees.

\begin{theorem}\label{thm:tandem-trees}
    Let \(\rho_n\) denote the number of distinct Manacher arrays corresponding to strings of length \(n\), and let $r_n$ denote the number of rooted tandem duplication trees with $n$ leaves. Then \(\rho_n \le r_{n+1}\).
\end{theorem}

Throughout the proof, we use an auxiliary combinatorial structure, denoted as the {\em counter array}. The array is defined as follows:

\begin{definition}[Counter array]\label{def:counter-array}
Let $A=(a_1,\dots,a_n)$ be an array of length~$n$ of integers
satisfying:
\[
  1\;\le\;a_i\;\le\; i
  \quad\text{and}\quad
  a_{i+1}\;\ge\;a_i-1
  \qquad(1\le i<n).
\]
We call $A$ a \emph{counter array} and denote by $\sigma_n$ the number
of such arrays.
\end{definition}

The proof process of~\cref{thm:tandem-trees} consists of three parts; 
In the first part of the proof (\cref{lem:dup-eq-count}), we prove that 
$\sigma_n=r_{n+1}$. In the second part (\cref{l:compact-man}), we show a compact representation of the Manacher array and prove that the Manacher array can be retrieved from it. In the third and last part (\cref{lem:compact-is-counter}), we show that the number of compact representations corresponding to strings of length $n$ is exactly $\sigma_n$, establishing $\rho_n \le \sigma_n = r_{n+1}$, as claimed.

\subsection{Rooted Duplication Trees}

In this subsection, we prove the following lemma:
\begin{lemma}[Duplication trees and counter arrays]\label{lem:dup-eq-count}
The number of rooted duplication trees with $n$ leaves equals the number of counter arrays of length $n$. That is,
\[
    r_{n+1} = \sigma_n.
\]
\end{lemma}

The proof is primarily combinatorial and requires several additional definitions related to rooted duplication trees. We refer to these trees simply as {\em duplication trees}.

Let $T$ be a duplication tree. All non-leaf nodes in $T$ have left and right children. We refer to the ordered list of genes at the leaves of the duplication as the {\em leaves array} of $T$. In~\cref{fig:rooted-tandem-repeat}, the leaves array of $T$ is the numbers from $1$ to $18$.

Given a duplication tree $T$ with leaves array $A=\{g_1,g_2,\dots,g_m\}$, a duplication event is a continuous indices array $B=\{i,i+1,\dots,i+\ell-1\}$. Applying the event $B$ to the array results in:
\[
B(A)=\{g_1,g_2,\dots,g_{i-1},lc(g_i),lc(g_{i+1}),lc(g_{i+\ell-1}),rc(g_{i}),\dots,rc(g_{i+\ell-1}),\dots,g_m\}
\]

Let $A = \{g_1,\dots,g_m\}$ be the leaves array.
We define a strict total order $\prec$ on the leaves by
$g_i \prec g_j \iff i < j$ for $1 \le i,j \le m$;
thus $g_1 \prec g_2 \prec \cdots \prec g_m$.
When needed, we write $g_i \preceq g_j$ to mean $i \le j$.

\textbf{Partial events ordering.} We define a partial order between two duplication events, $B_1$ and $B_2$. We say that $B_1$ is smaller than $B_2$ if $\max B_1 <\min B_2$. Note that the order is not necessarily defined; let $B_1=\{1,2\},B_2=\{2,3\}$, then $B_1\nprec B_2$ and $B_2\nprec B_1$.  However, if $B_1=\{6\},B_2=\{1,2,3\}$, then $B_2\prec B_1$. We say that $B_1 \preceq B_2$ if and only if $B_2 \nprec B_1$. The relation $\preceq$ is defined between every pair of events.

We can now prove a key lemma that also results in an algorithm to decompose a duplication tree into a unique list of duplication events.

\begin{lemma}[Unique event decomposition]\label{lem:decompose}
Any rooted duplication tree with \(n\) leaves admits a \emph{unique}
ordered list \((B_1\preceq B_2 \preceq\dots \preceq B_k)\) of duplication events
such that applying the events in that order recreates the tree’s leaves.
\end{lemma}
\begin{proof}\label{ap:pr:decompose}
We begin by defining {\em low nodes} for the proof.

\begin{definition}[Low nodes]
A node in the tree is called {\em low} if and only if all of its children are leaves. For two low nodes $u,v$ write $u\prec v$ iff
$\lc(u)\prec\lc(v)$ in the leaf order.
\end{definition}

We induct on $n$. For $n=1$ the tree consists of the original gene, and no duplication is performed.

\medskip\noindent
\emph{Induction step.}
Assume every tree with $n-1$ leaves has a unique decomposition, and let
$T$ have $n$ leaves.  Call a node \emph{low} if both its children are
leaves; let $\mathcal U$ be the set of low nodes of $T$.

\paragraph{Identifying the last event.}

Low nodes created by the \emph{same} duplication event are precisely those
whose child leaves interleave:
\[
  \lc(u)\prec\lc(v)\prec\rc(u)\prec\rc(v)
  \quad\text{or}\quad
  \lc(v)\prec\lc(u)\prec\rc(v)\prec\rc(u).
\]
Partition $\mathcal U$ into blocks according to this interleaving rule.
Each block is a subset of one duplication event; if the leaves generated by a block are not contiguous in the leaves array, the block cannot be the \emph{last} event and is discarded.
Of the blocks that remain, choose the one with the maximal low node; call it $B^{\ast}$.
Its leaves occupy a contiguous interval
$\{i,\dots,i+2\ell-1\}$, so the event itself is
$B^{\ast}=\{i,\dots,i+\ell-1\}$.

\paragraph{Remove the last event.}
Delete the $\ell$ right‑copy leaves of $B^{\ast}$ and suppress the resulting
degree‑1 parents.  The resulting tree $\widetilde T$ has $n-\ell$ leaves.
By the induction hypothesis, $\widetilde T$ decomposes uniquely as
$(B_1,\dots,B_k)$.

\paragraph{Ordering.}
It remains to show $B_k\preceq B^{\ast}$, ensuring the combined list
\[ (B_1,\dots,B_k,B^{\ast}) \]
is totally ordered.
There are two cases:

1.  If $B_k$ was already a candidate block in $T$, $B_k\prec B^{\ast}$
    because $B^{\ast}$ was chosen as the highest candidate.

2.  Otherwise, $B_k$ contains a leaf that stopped being low only after
    $B^{\ast}$ was removed; that leaf lies in $B^{\ast}$.
    Hence $\min B_k\le\max B^{\ast}$, i.e.\ $B_{m-1}\preceq B^{\ast}$.

\noindent
Thus $(B_1,\dots,B_k,B^{\ast})$ is a valid ordered decomposition of $T$.
Uniqueness follows because the construction of $B^{\ast}$ and the induction
hypothesis are unique. \qed

\begin{example}
    In~\cref{fig:rooted-tandem-repeat}, the set of low nodes is $\{i,q,n,j,k,o,p,m\}$. The partition of this set into duplication subsets is $\{\{i\},\{q\},\{n\},\{j,k\},\{o\},\{p\},\{m\}\}$, and the set of leaves generated by these subsets is:
    \[
    \{\{2,5\},\{3,4\},\{6,7\},\{8,9,10,11\},\{12,13\},\{14,15\},\{16,18\}\}
    \]
    It can be seen, that the leaves corresponding to $\{i\}$ and $\{m\}$ do not form a continuous list, and therefore are discarded, and we are left with the duplication candidates $\{\{q\},\{n\},\{j,k\},\{o\},\{p\}\}$. The highest duplication event is $\{p\}$, which corresponds to $\{14,15\}$, and therefore the last duplication event in this example is $\{14\}$, and the previous leaves array $\wit{A}$ is:
    \[(1,2,3,4,5,6,7,8,9,10,11,12,13,p,16,17,18)\]
    By applying this procedure again, the duplication event will be the one that created $(p,16,17,18)$, which is $\{14,15\}$, and $\{14,15\}\preceq \{14\}$.
\end{example}

\end{proof}







Next, we prove that there is a bijection between rooted duplication trees and counter arrays, as stated at the beginning of the subsection at~\cref{lem:dup-eq-count}.
    
\begin{proof}
    \textit{Encoding events as a counter array.}
    Let $T$ be a duplication tree with events
    $B_1\preceq\dots\preceq B_{k}$ and
    $B_i=(b_i,\dots,b_i+\ell_i-1)$.
    Encode $B_i$ by the strictly decreasing sequence
    $(b_i+\ell_i-1,\dots,b_i)$ and concatenate the
    $k$ sequences to obtain $A=(a_1,\dots,a_n)$.
    Inside an event, the drop is exactly 1; between
    events, we have $a_{j+1}\ge a_j$ \footnote{Because the events are ordered.}, 
    so $a_{j+1}\ge a_j-1$ for all $j$.
    While encoding $B_i$ there are exactly
    $1+\sum_{t<i}\ell_t$ genes, hence every
    symbol written satisfies $1\le a_j\le j$.
    Thus, $A$ is a counter array, and its length is $n-1=\sum_{i=1}^k\ell_i$.
    
    \smallskip
    \textit{Decoding a counter array.}
    Conversely, scan a counter array $A$ from left to right.
    Start a new event whenever the next entry fails to drop by 1. If entry $a_j$ starts a new event,
    the number of available genes is exactly $j$. Therefore, an entry that starts
    an event and the subsequent entries within the event are valid gene indices when
    read, hence every recovered event is legal.
    The two procedures are inverses, establishing a bijection between
    rooted duplication trees with $n$ genes and counter arrays of length $n-1$.
    Hence $r_{n+1}=\sigma_n$. \qed
\end{proof}

\subsection{Compact Manacher Array}

At the beginning of the section, we proved that a Manacher array requires $\Omega(n)$ bits to represent. We show a folklore representation of the Manacher array that requires $\Theta(n)$ bits\footnote{Arseny M. Shur, private communication}. Later, we show that this representation is equivalent to the counter array. We conclude by showing that the representation is not tight --- meaning that some counter arrays do not represent a valid Manacher array.

\begin{lemma}\label{l:compact-man}
    The Manacher array $A$ of a string $S$ of length $n$ can be represented using $\Ohn$ bits.
\end{lemma}

\textit{High-level idea.} 
    For each position $i$ let $c_i$ be the center of the maximal palindromic suffix of $S[1..i]$.
    Define $b_1:=1$ and $b_i:=2(c_i-c_{i-1})$ for $i\ge2$.  
    Because each suffix extends the previous one or ends earlier, the sequence
    $C=(c_i)$ is non‑decreasing; hence $\mathsf{B}=(b_i)$ is a monotone sequence of
    non‑negative integers.

    Since $c_i\le n$ and $2c_i = \sum_{j\le i} b_j$,
    \[
    \sum_{i=1}^{n} b_i \le 2n .
    \]
    Encode $\mathsf{B}$ in unary, separating consecutive values by a single 0.  This uses at
    most $(n-1)$ zeros and $\le 2n$ ones, i.e.\ $\le 3n-1 = \Ohn$ bits.

    From the unary code we recover $\mathsf{B}$, prefix‑sum to obtain each $c_i$, and hence
    every center.  A right‑to‑left sweep (mirroring the standard Manacher update)
    then assigns the correct radius to each center in $\Ohn$ time.

And now we proceed to show the full proof.
\begin{proof}\label{ap:pr:compact-man}
    Let $s_i$ be the maximal palindromic suffix of $S[1..i]$, and let $c_i$ be the center of the maximal palindromic suffix. Note that $c_i=i-\frac{|s_i|-1}{2}$, and $c_1=1$.

    The compact representation is a non-decreasing array $B=(b_1,\dots,b_n)$, where $b_1:=1$, and $b_i:=2(c_i-c_{i-1})$. We need to show that:
    \begin{enumerate}
        \item The representation takes $O(n)$ bits.
        \item The Manacher array $\mathsf{A}$ can be reconstructed from the compact representation $B$.
    \end{enumerate}

    Note that we multiply $c_i-c_{i-1}$ by two, since the centers might be half-integers, and we want the array to contain only integers. 

    \paragraph{Space.} First, observe that the centers $(c_i)$ form a non‑decreasing sequence.
    Indeed, suppose $c_{i+1} < c_i$.  Then the maximal palindromic suffix $s_{i+1}$ of $S[1..i+1]$ extends two positions beyond $s_i$; deleting its first and last characters yields a palindromic suffix of $S[1..i]$ that is longer than $s_i$, contradicting the maximality of $s_i$.

    Proceeding to the proof, each center $c_i$ satisfies $c_i\le n$. Consequently, since $2c_i=\sum_{j=1}^i b_j$ and the maximal value of some $c_i$ is $n$, the total sum of elements in $B$ is at most $2n$. 
    
    Since the elements $b_i$ are non-decreasing, we can represent each value of $b_i$ in unary form, and separate values with a zero. The number of zeroes is exactly $n-1$, and the number of ones is at most $2n$. Therefore, the compact representation can be performed using a binary string of length at most $3n-1=\Oh{n}$ bits, as required.

    \paragraph{Reconstruction.}
    Let $B$ be a compact representation of length $n{+}1$; the corresponding
    Manacher array has length $2(n{+}1)-1=2n{+}1$.
    Assume inductively that any compact representation of length $n$
    recovers its $(2n{-}1)$‑entry Manacher array.
    
    \smallskip
    \textit{Base case.}
    For $|B|=1$ the string has length 1, whose single center has radius~0.
    
    \smallskip
    \textit{Inductive step.}
    Let $\mathsf A'$ be the $(2n{-}1)$‑entry array reconstructed from
    $B[1..n]$ by the induction hypothesis.
    The new character adds two centers
    $c=n+\tfrac12$ and $c=n{+}1$, which we initialize with radius 0.
    Compute
    \[
      c_{n+1}=\tfrac12\sum_{i=1}^{n+1} b_i,
      \qquad
      r_{n+1}= \bigl\lceil (n{+}1)-c_{n+1}\bigr\rceil ,
    \]
    and set $\mathsf A'[c_{n+1}]=r_{n+1}$ (this is the new maximal suffix
    palindrome).
    Every existing center whose palindrome now reaches position $n{+}1$ needs its radius increased to
    $r_c=\lceil(n{+}1)-c\rceil$.
    Those centers satisfy $c>c_{n+1}$.
    Their mirrored partners
    $c^\ast=2c_{n+1}-c$ lie to the left of $c_{n+1}$ and \emph{do not}
    reach $n{+}1$, so $\mathsf A'[c^\ast]$ is already correct.
    Update
    \[
        \mathsf A'[c] \;=\;
            \min\bigl(r_c,\; \mathsf A'[c^\ast]\bigr)
            \quad\text{for all integer centers }c\in(c_{n+1},\,n{+}1].
    \]
    
    This completes the reconstruction for length $n{+}1$ and, by induction,
    for all lengths, proving the lemma. \qed
\end{proof}



\begin{lemma}\label{lem:compact-is-counter}
    The compact representation as defined in~\cref{l:compact-man} is equivalent to the counter array.
\end{lemma}

\begin{proof}\label{ap:pr:compact-is-counter}
    Let $\mathsf{B}$ be the compact representation of length $n$. Define $\mathsf{B}'$ as the prefix-sums array of $\mathsf{B}$, i.e., $\mathsf{B}'[i]=\sum_{j=1}^i\mathsf{B}[j]$. The values in $\mathsf{B}'$ satisfy $\mathsf{B}'[i]=2c_i$, where $c_i$ is the center of the maximal palindromic prefix ending at index $i$. Denote $b_i=\mathsf{B}'[i]$.

    The minimal value for the center $c_i$ is $\frac{i+1}{2}$, and the maximal value is $i$. Therefore, $(i+1)\le b_i\le 2i$. Additionally, since the centers form a non-decreasing sequence, $b_i\le b_{i+1}$. 

    Consider the array $\wit{B}=(\tilde{b}_i)$, where $\tilde{b}_i=b_i-i$. The following holds:
    \begin{enumerate}
        \item $1\le \tilde{b}_i \le i$
        \item $b_i\le b_{i+1}\to b_i-i\le b_{i+1}-(i+1)\to b_{i+1}-b_i\ge -1$
    \end{enumerate}

    Those are the exact requirements of the counter array. \qed
\end{proof}

\begin{corollary}
    The value $\rho_n$ is less than or equal to the number of distinct compact representations from~\cref{l:compact-man}.
\end{corollary}

And with that, we conclude the proof of~\cref{thm:tandem-trees}. \qed

\section{Restriction Graphs}\label{sec:graph}

In this section, we prove the following theorem:

\begin{theorem}\label{thm:graph-reconstruction}
    Let $\man$ be a valid Manacher array. There exists a graph $G=(V,E)$, such that every proper coloring of $G$ with $k$ colors yields a string with Manacher array $\man$ and exactly $k$ different alphabet symbols, and every string with Manacher array $\man$ yields a proper graph coloring.
\end{theorem}

A similar framework considering only equality dependency relations in strings was presented by Gawrychowski et al.~\cite{GKRR:20}, where they seek the biggest possible alphabet to satisfy the given set of dependencies. Generally, reconstructing a string from a general set of dependencies using the smallest possible alphabet is NP-hard~\cref{c:colnp}. Our approach incorporates both equality and inequality dependencies, but restricts the dependencies to a specific, non-arbitrary set, thus avoiding this complexity. 

Throughout the proof, we refer to a variant of the Manacher array of a string $S$ as the {\em palindromic fingerprint}, or simply {\em fingerprint} of $S$.  

\begin{definition}\label{def:palf}
    The {\emph palindromic fingerprint} $\fp$ of a string $S$ is a set of all maximal palindromes in $S$.
    
    A pair $(i, j)$ is included in the fingerprint $\fp$ if and only if $S[i..j]$ is a maximal palindrome. Zero-length maximal palindromes are $(i, i-1)$.
    
    The length of a fingerprint $\fp$, denoted $|\fp|$, is equivalent to the length of the underlying string $|S|$.

    A string $T$ is a {\em reconstruction} of $\fp$ if the set of maximal palindromes of $T$ results in the set $\fp$, and $|T|=|\fp|$.
\end{definition}

For example, denote $\fp$ as the fingerprint of string $S$ with length $12$ and a maximal palindrome at $S[2..7]$. Then $(2,7)\in \fp$ and $|\fp|=12$. Additionally, $S$ is a reconstruction of $\fp$.

We present the {\em restriction graph} $G$ of a palindromic fingerprint $\fp$, which is the graph that is described at~\cref{thm:graph-reconstruction}.

The equivalence between restriction graph coloring and palindromic fingerprint reconstruction gives us a straightforward bound on the number of distinct alphabet characters required to reconstruct a given fingerprint - the smallest possible number is the chromatic number of the graph $\chi(G)$. The highest possible number is the number of nodes in $G$, $|V|$.

In this section, we provide a detailed description of the restriction graph. First, we define an auxiliary graph, called the {\em equality graph}.

\begin{definition}[Equality Graph]\label{d:eqg}
The equality graph of a fingerprint $\fp$ of length $n$ denoted as $\eqf:=(V, E)$ is an undirected graph defined as follows:\\
$V=\{i\}_{i=1}^n$, and $(i, j) \in E$ if and only if $(i,j)$ are palindromically dependent (recall from~\cref{def:pal-dep}: If $i,j$ are palindromically dependent, then in every reconstruction $T$ of $\fp$, $T[i]=T[j]$).
\end{definition}

The following follows from the definition of $G_{=}(\fp)$.
\begin{observation}
Let $A_1,A_2, \ldots ,A_\ell$ be the connected components of $\eqf$. Every reconstruction $T$ of $\fp$ satisfies $S[i] = S[j]$ if and only if: 
\[
\exists k\quad\text{s.t.}\quad i,j \in A_k
\]
\end{observation}

We now define the restriction graph of $\fp$.
\begin{definition}[Restriction Graph]\label{d:rsg}
Let $G' = \eqf = (V', E')$ be the \textit{equality graph} of $\fp$.

The \textbf{restriction graph} $G = G(\fp) :=(V, E)$ is an undirected graph defined as follows:
\begin{itemize}
    \item The vertex set $V$ consists of the connected components of $G'$, i.e.,  
    \[
    V = \{A_k\}_{k=1}^{\ell}
    \]
    
    \item The edge set $E$ consists of edges between connected components that are constrained by palindromic properties in $\fp$. 

    Informally, if in every possible reconstruction $T$ of $\fp$, $T[i_1]\neq T[i_2]$, then we draw an edge between the vertex $A_k^1$ that contains $i_1$ and the vertex $A_k^2$ that contains $i_2$.
    
    Formally,  
    \[
    E = \{(A_{k_1}, A_{k_2}) ~|~ \exists i \in A_{k_1}, \exists j \in A_{k_2} \text{ such that } (i+1, j-1) \in \fp \}.
    \]
\end{itemize}
\end{definition}

A detailed example of the restriction graph and its construction can be found at~\cref{e:rsg}.

The following two observations can hint to the relation between graph coloring and fingerprint reconstruction:
\begin{observation}\label{ob:all-par}
    Let $G=(V,E)$ be a restriction graph for a palindromic fingerprint $\fp$. The vertices $V=(v_1,v_2,...)$ form a partition of $[n]=\{1,2,\ldots,n\}$.
\end{observation}

\begin{observation}\label{o:neqres}
Let $T$ be a reconstruction of $\fp$, and let $G=(V,E)$ be the restriction graph of $\fp$. 
If $(A_{k_1}, A_{k_2})\in E$, then:
\[
\forall i\in A_{k_1}, \forall j\in A_{k_2}\qquad T[i] \neq T[j].
\]
\end{observation}

We proceed to prove the equivalence ``coloring $\iff$ reconstruction''.

\begin{lemma}\label{l:ctos}
    Let $\psi:V\to \Sigma$ be a coloring of the restriction graph $G(\fp)=(V,E)$ with $n=|\fp|$. 

    Define another function $\psi':[n]\to \Sigma$:
    \[
    \forall A_k\in V\quad \forall i\in A_k\qquad \psi'(i):=\psi(A_k)
    \]
    
    A string $T$ can be reconstructed as  
    \[
    T = \psi'(1) \psi'(2) \ldots \psi'(n),
    \]
    And the fingerprint of $T$ is $\fp$.

    Conversely, given a string $T$ with fingerprint $\fp$, define a coloring function $\psi$:
    \[
    \psi(A_k):=T[i] \quad \text{For some } i\in A_k
    \]
    The function $\psi$ is a proper coloring of $G(\fp)$.

\end{lemma}
\begin{proof}\label{ap:pr:ctos}
    We establish a bijection between colorings of $G(F)$ and strings with fingerprint $F$.  
    
    \textbf{(Coloring to String)}  
    First, due to~\cref{ob:all-par}, every index $1 \le i\le n$ has a unique set $A_k$ such that $i\in A_k$ and $A_k\in V$, so $\psi'$ is well defined on all values of $i$.

    Consider the resulting string:
    \[
    T = \psi'(1) \psi'(2) \dots \psi'(n)
    \]
    The string $T$ is well defined. Consider two palindromically dependent indices $i,j$. Since $i,j$ are palindromically dependent, they belong to the same connected component in the equality graph of $\fp$, and thus, to the same vertex in the restriction graph. Therefore, by our definition of $\psi'$, we know that $\psi'(i)=\psi'(j)$, and therefore, $T[i]=T[j]$.

    Now, assume two indices $i,j$ must not be equal. Such a restriction can be imposed by a maximal palindrome $(i+1,j-1)\in \fp$. However, by the definition of $G$, such a restriction implies that the vertex containing $i$ and the vertex containing $j$ are connected by a vertex, and therefore $\psi'(i)\neq \psi'(j)$, and $T[i]\neq T[j]$.
    
    Overall, we guaranteed that every maximal palindrome remains unchanged in $T$, preserving $\fp(T) = \fp$.
    
    \textbf{(String to Coloring)}  
    Conversely, let $T$ be a valid reconstruction of $\fp$. Define $\psi: V \to \Sigma$ as  
    \[
    \psi(A_k):=T[i] \quad \text{For some } i\in A_k
    \]

    First, since the sets $(A_k)$ are the connected components of the restriction graph, $\forall i,j\in A_k\quad T[i]=T[j]$. Therefore, the function $\psi$ is unique.

    By Observation~\ref{o:neqres}, any two connected components in the restriction graph are assigned distinct symbols in $T$, making $\psi$ a valid coloring for $G$.
    
    Thus, the mappings are inverses, establishing the bijection. \qed

\end{proof}


And the latter proof completes the proof for~\cref{thm:graph-reconstruction}. \qed

        
        


The following is a detailed example of the restriction graph and its construction.
\begin{example}[Restriction graph - ~\cref{d:rsg}]\label{e:rsg}

    Consider the string $S=\texttt{41213121566757}$.

    The maximal palindromes in the string are: 
    \[\{(2,4),(2,8),(6,8),(10,11),(12,14)\}\]
    Trivial palindromes of length one and zero are discarded from the set.

    The equality graph, $G_{=}$, is a graph where each index in the original string is assigned to a node. In the graph, we only connect nodes that ought to be connected by a palindromic restriction. The equality graph is in~\cref{f:eqg}.

    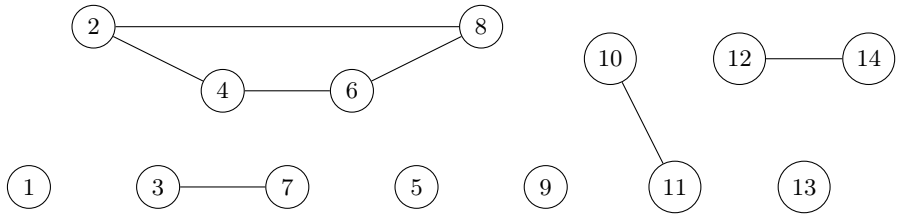
\begin{figure}[ht]
    \caption{The equality graph of $S$ }

    \begin{tikzpicture}[scale=0.85, node distance=1.5cm, every node/.style={draw, circle}]
    \node (v1) at (0,0) {1};

    \node (v2) at (1,2.5) {2};
    \node (v4) at (3,1.5) {4};
    \node (v6) at (5,1.5) {6};
    \node (v8) at (7,2.5) {8};

    \node (v3) at (2,0) {3};
    \node (v7) at (4,0) {7};

    \node (v5) at (6,0) {5};
    \node (v9) at (8,0) {9};
    \node (v10) at (9,2) {10};
    \node (v11) at (10,0) {11};
    \node (v12) at (11,2) {12};
    \node (v13) at (12,0) {13};
    \node (v14) at (13,2) {14};
    
    \draw (v2) -- (v4);
    \draw (v4) -- (v6);
    \draw (v6) -- (v8);
    \draw (v2) -- (v8);
    \draw (v7) -- (v3);
    \draw (v12) -- (v14);
    \draw (v10) -- (v11);
    \end{tikzpicture}

    \label{f:eqg}
    \end{figure}

    The restriction graph $G$ shrinks all connected components to only one node. We labeled the nodes with their original indices, separated by commas. In the restriction graph, we draw an edge between two nodes that ought to be reconstructed using another label. Note that although we omitted trivial palindromes in the fingerprint, their existence can be inferred from the lack of other palindromes in the given center. When considering the restriction graph, we also consider restrictions of trivial palindromes. For example, we can see that there is no palindrome centered at 10, which means that the palindrome centered at 10 is trivial, hence $S[9] \neq S[11]$. The restriction graph is in~\cref{f:restg}.

\begin{figure}[ht]
    \centering

    \begin{subfigure}{0.45\textwidth}
        \centering
        \caption{The restriction graph $G$ of $S$. The biggest clique is highlighted in red.}
        \begin{tikzpicture}[node distance=1.5cm, every node/.style={draw, circle, minimum size=1cm, inner sep=0pt}]
            \foreach \i/\j/\label in {1/1/1, 3/2/{3,7}, 2/3/{2,4,6,8}, 10/4/{10,11}, 5/5/5, 12/6/{12,14}, 9/7/9, 13/8/13} {
                \node (v\i) at ({360/8 * (\j - 1)}:2.75) {\label};
            }
            
            \foreach \u/\v in {1/3, 1/2, 1/5, 1/9, 3/2, 3/5, 3/9, 2/5, 2/9, 5/9} {
                \draw[red] (v\u) -- (v\v);
            }
            
            \foreach \u/\v in {12/13, 12/10, 10/13, 9/10, 10/2} {
                \draw (v\u) -- (v\v);
            }
        \end{tikzpicture}
        \label{f:restg}
    \end{subfigure}
    \hfill
    \begin{subfigure}{0.45\textwidth}
        \centering
        \caption{Optimal coloring of $G$ with 5 colors.}
        \begin{tikzpicture}[node distance=1.5cm, every node/.style={draw, circle, minimum size=1cm, inner sep=0pt, very thick}]
            \tikzmath{\r = 2.75;}
            \node[draw=blue] (v1) at ({360/8 * (1 - 1)}:\r) {1};
            \node[draw=red] (v3) at ({360/8 * (2 - 1)}:\r) {3,7};
            \node[draw=green] (v2) at ({360/8 * (3 - 1)}:\r) {2,4,6,8};
            \node[draw=purple] (v10) at ({360/8 * (4 - 1)}:\r) {10,11};
            \node[draw=purple] (v5) at ({360/8 * (5 - 1)}:\r) {5};
            \node[draw=red] (v12) at ({360/8 * (6 - 1)}:\r) {12,14};
            \node[draw=magenta] (v9) at ({360/8 * (7 - 1)}:\r) {9};
            \node[draw=blue] (v13) at ({360/8 * (8 - 1)}:\r) {13};
            
            \foreach \u/\v in {1/3, 1/2, 1/5, 1/9, 3/2, 3/5, 3/9, 2/5, 2/9, 5/9} {
                \draw (v\u) -- (v\v);
            }
            
            \foreach \u/\v in {12/13, 12/10, 10/13, 9/10, 10/2} {
                \draw (v\u) -- (v\v);
            }
        \end{tikzpicture}
        \label{f:optc}
    \end{subfigure}

     \caption{Illustration of the restriction graph and its optimal coloring.}
\end{figure}

    Our last step towards reconstruction is coloring the graph. We can see that the subgraph with nodes $\{(1), (2,4,6,8), (3,7), (5), (9)\}$ is the complete graph $K_5$, hence there is no coloring with less than five colors. Also, there are only eight nodes, which implies that there is no coloring with more than eight colors.

    Our original string $S$, is a valid coloring using seven colors. \\
    The string $S=\texttt{41213121566787}$ is the naïve coloring that assigns each node a different color. The string $S=\texttt{45253525133212}$ represents an optimal coloring of the graph, e.g., coloring with the fewest possible colors. The coloring is presented in~\cref{f:optc}.
    
\end{example}

\section{Alphabet Size Bounds}\label{sec:alph}

Logarithmic bounds on the alphabet size in palindromic-equivalent structures are well established. However, a finer combinatorial analysis reveals even more specific structural constraints. In this section, we present results that clarify the relationship between a fingerprint's alphabet size and its underlying combinatorial structure.

\begin{theorem}\label{thm:minimal-reconstruction}
    Let $\man$ be a Manacher array, and let $S$ be a lexicographically minimal string with the Manacher array $\man$, i.e., a {\em reconstruction} of $\man$. Assume the alphabet of $S$ is $\Sigma=\{\sigma_1,\sigma_2,\dots\}$, and that the lexicographic ordering is $\sigma_i<\sigma_{i+1}$ for every $i$.

    \begin{enumerate}
    \item  The first occurrence of a character \( \sigma_i,\; i\ge 3 \) is preceded by a substring that follows the pattern \( \ZP{i-2} \).  
    \item  Any subsequent occurrences of a character \( \sigma_i \) (where \( i \geq 3 \)) are either:
        \begin{itemize}
           \item  At a palindromically dependent index (recall~\cref{def:pal-dep}), or  
           \item  Preceded by a substring that follows the pattern \( \ZP{i-2} \).  
        \end{itemize}
        \item  The string \( S \) is constructed using a globally minimal alphabet.
    \end{enumerate}
\end{theorem}

The above theorem implies the following:

\begin{corollary}[Alphabet Size]\label{cor:alph-size}
    Let $\alpha:\mathbb{N} \to \mathbb{N}$ be the function such that $2\alpha(k) - 1$ is the minimal length of a Manacher array that cannot be realized by any string over an alphabet of fewer than $k$ distinct symbols.
    
    Then,
    \[
    \alpha(k) = 
    \begin{cases} 
        1, & \text{if } k = 1, \\ 
        2^{k-2} + 1, & \text{if } k \ge 2.
    \end{cases}
    \]
\end{corollary}

The corollary is achieved by plugging the shortest string that matches $\ZP{k}$ into the theorem. An example of the corollary can be found at~\cref{ap:ex:alph}.\\
We acknowledge that~\cite{tomohiro:10} claims their algorithm yields a minimal alphabet; however, this claim is not rigorously proven in their paper.

\begin{example}[Tight alphabet size example for~\cref{cor:alph-size}]\label{ap:ex:alph}
    We now demonstrate the shortest string that follows the pattern $\ZP{k}$. Since $\ZP{k}$ is a palindrome of the form $\ZP{k-1}p_k\ZP{k-1}$, and $p_k$ is an arbitrary palindrome, we are going to choose all palindromes $p_1,p_2,\dots,p_k$ to be of length 1. Therefore, we denote $\sigma_i=p_i[1]=p_i$. Recall that $i\neq j\to p_i\neq p_j$.

    Therefore, let $P_k$ be the minimal string that follows $\ZP{k}$, then:
    \[
    P_1=1\quad P_2=121 \quad P_3=1213121 \quad P_k=P_{k-1}P_kP_{k-1}
    \]

    And by~\cref{thm:minimal-reconstruction}, the minimal Manacher array that requires at least $k$ characters to reconstruct, is the Manacher array of:
    \[
    S_k=\sigma_{k-1}P_{k-2}\sigma_{k}
    \]

    And by setting $k=5$, we obtain:
    \[
    S_5=412131215
    \]
    And $|S_5|=9=2^{5-2}+1$, as required.
\end{example}

\noindent \textbf{Proof of~\cref{thm:minimal-reconstruction}} Throughout the proof, let $\man$ denote the given Manacher array, and let $S$ be the lexicographically minimal string corresponding to $\man$. We assume that $S$ contains more than three distinct characters, i.e., $|\Sigma_S| \geq 4$, where $\Sigma_S$ represents the alphabet of $S$. We fix $k$ to be the alphabet size, i.e., $k = |\Sigma_S|$.

\begin{observation}\label{o:pal-set}
    Let $\sigma_i$ be a character at a palindromically-independent index $m$, and let $s_m=S[1..m]$. For every character $\sigma_j$ with $j<i$, there exists a palindrome $p_j$ such that $\sigma_jp_j\sigma_i$ is a suffix of $s_m$.
\end{observation}

This observation is an alternative formulation of the requirement that the string $S$ is lexicographically minimal.

\begin{observation}\label{o:per-pal-struc}
    Let $P$ be a periodic palindrome. $P$ can be written as $(q_0q_1)^iq_0$, where $q_0$ and $q_1$ are palindromes and $i\geq 2$.
\end{observation}

A demonstration of this observation can be found at Figure~\ref{fig:per-pal-struct}.

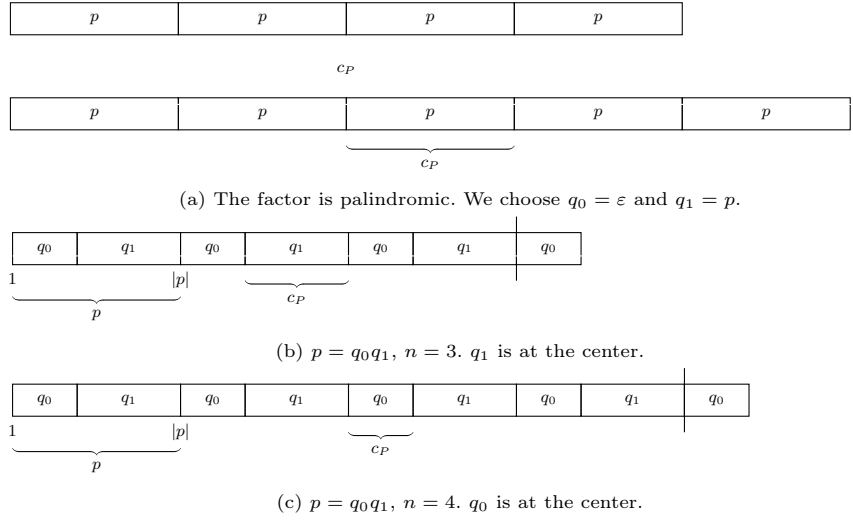
\begin{figure}[ht]
\centering

\newcommand{\rheight}{0.5}
\newcommand{\Alength}{14}
\newcommand{\qzlength}{1.0}
\newcommand{\qolength}{1.6}
\pgfmathsetmacro{\tlength}{\qolength+\qzlength}

\newcommand{\palper}[1]{
\pgfmathsetmacro{\n}{#1}
\pgfmathsetmacro{\nm}{#1-1}
\pgfmathsetmacro{\tsize}{\tlength*\n}

\node (ha) at (0,0) {};

\pgfmathsetmacro{\startX}{iseven(\n) ? (\n/2)*(\tlength) : ((\n+1)/2)*\tlength }
\pgfmathsetmacro{\centerLength}{iseven(\n) ? \tlength : 0}

\draw [decorate, orange, decoration = {calligraphic brace,mirror}] (\startX,-0.5*\rheight) -- ({\startX+\centerLength},-0.5*\rheight) node[midway,shift={(0,{-0.7*\rheight})},black] {$c_P$};

\node (ps) at (0,0) {};
\node (pe) at (2.6,0) {};

\foreach \x in {0,...,\n} {
    \draw ($(\x*\tlength,0)$) rectangle ++ (\tlength,\rheight) node[midway] {$p$};
}

\draw[white] (0,5*\rheight) -- (\Alength,5*\rheight);
}

\newcommand{\nonpalper}[1]{
\pgfmathsetmacro{\n}{#1}
\pgfmathsetmacro{\nm}{#1-1}
\pgfmathsetmacro{\tsize}{\tlength*\n}

\node (ha) at (0,0) {};

\pgfmathsetmacro{\startX}{iseven(\n) ? (\n/2)*(\tlength) : (\nm/2)*\tlength + \qzlength }
\pgfmathsetmacro{\centerLength}{iseven(\n) ? \qzlength : \qolength }

\draw [decorate, orange, decoration = {calligraphic brace,mirror}] (\startX,-0.5*\rheight) -- ({\startX+\centerLength},-0.5*\rheight) node[midway,shift={(0,{-0.7*\rheight})},black] {$c_P$};

\draw (\tsize,-0.5*\rheight) -- (\tsize,1.5*\rheight);

\node (ps) at (0,0) {};
\node (pe) at (2.6,0) {};

\node[below,text height=0.25cm]  at ($(ps)$) {$1$};
\foreach \x in {0,...,\nm} {
    \draw ($(\x*\tlength,0)$) rectangle ++ (\qzlength,\rheight) node[midway] {$q_0$};
    \draw ($(\x*\tlength+\qzlength,0)$) rectangle ++ (\qolength,\rheight) node[midway] {$q_1$};
}
\draw ($(\n*\tlength,0)$) rectangle ++ (\qzlength,\rheight) node[midway] {$q_0$};

\node[below,text height=0.25cm]  at ($(pe)$) {$|p|$};

\draw [decorate,
    decoration = {calligraphic brace,mirror}] ($(ps)+(0,-\rheight)$) --  ($(pe)+(0,-\rheight)$) node[midway,shift={(0,{-0.7*\rheight})}] {$p$};

\draw[white] (0,5*\rheight) -- (\Alength,5*\rheight);
}


\setlength{\belowcaptionskip}{1pt}

\begin{subfigure}{\textwidth}
\centering
\begin{tikzpicture}[scale=0.85, every node/.style={scale=0.7}]
\palper{3}
\end{tikzpicture}
\end{subfigure}

\vspace{-1.5cm} 

\begin{subfigure}{\textwidth}
\centering
\begin{tikzpicture}[scale=0.85, every node/.style={scale=0.7}]
\palper{4}

\end{tikzpicture}
\caption{The factor is palindromic. We choose $q_0=\varepsilon$ and $q_1=p$.}
\end{subfigure}

\vspace{-1.5cm} 

\begin{subfigure}{\textwidth}
\centering
\begin{tikzpicture}[scale=0.85, every node/.style={scale=0.7}]
\nonpalper{3}

\end{tikzpicture}
\caption{$p=q_0q_1$, $n=3$. $q_1$ is at the center.}
\end{subfigure}

\vspace{-1.5cm} 

\begin{subfigure}{\textwidth}
\centering
\begin{tikzpicture}[scale=0.85, every node/.style={scale=0.7}]
\nonpalper{4}
\end{tikzpicture}
\caption{$p=q_0q_1$, $n=4$. $q_0$ is at the center.}
\end{subfigure}

\caption{Demonstration of~\Cref{o:per-pal-struc}. $c_P$ is the center of the periodic palindrome.}
\label{fig:per-pal-struct}
\end{figure}


\begin{lemma}\label{l:expo}
    Let \( \sigma_i \) be a character at a palindromically independent index \( m \), and let \( s_m = S[1..m] \). Let \( \mathcal{P} = \{p_1, p_2, \ldots, p_{i-1}\} \) be the set of palindromes as defined in~\cref{o:pal-set}, where each $p\in \mathcal{P}$ is the shortest possible. For any distinct pair of palindromes \( p_j \) and \( p_{j'} \) from \( \mathcal{P} \), where \( |p_j| < |p_{j'}| \), it holds that \( 2|p_j| < |p_{j'}| \). The longest palindrome $p_j$ matches the pattern $\ZP{i-2}$.
\end{lemma}

\begin{proof}
Let us assume without loss of generality that the set $\mathcal{P}=\{p_1,p_2,\ldots,p_{i-1}\}$ is sorted, i.e., for every pair $p_j,p_{j+1}$ it holds that $|p_j|>|p_{j+1}|$.

Assume to the contrary there exists a triplet $(\sigma_i,p_j,p_{j+1})$ that contradicts the lemma, meaning:
\begin{enumerate}\label{temp:assum}
    \item $2|p_j|\ge |p_{j+1}|$
    \item $\sigma_jp_j\sigma_i$ is a suffix of $s_m$. The same condition would hold for $\sigma_{j+1}$ and $p_{j+1}$, respectively.
\end{enumerate}

First, let's rule out the possibility of $2|p_j|=|p_{j+1}|$. If $2|p_j|=|p_{j+1}|$, then $p_{j+1}[|p_j|]=p_{j+1}[|p_j|+1]=\sigma_j$, and by symmetry of $p_j$, the last character of $p_j$ also equals $\sigma_j$, and therefore the minimal palindrome $p_j$ is the empty string $\varepsilon$. However, we assumed $2|p_j|=|p_{j+1}|=0$, contradiction.

Let us now assume $2|p_j|>|p_{j+1}|$. When two palindromic suffixes are of similar length, the longer palindrome has a period of length at most $|p_{j+1}|-|p_j|$. Therefore, let $q$ be the minimal periodic factor, and rewrite $p_j$ as $(q_0q_1)^iq_0$~\cref{o:per-pal-struc}.

We know that:
\begin{enumerate}
    \item The character preceding $p_j$ in $S'$ is $\sigma_j$.
    \item $|p_{j+1}|\geq |p_{j}|+1$, and therefore $a_jp_j$ is a suffix of $p_{j+1}$, which is a periodic string with factor $q$.
\end{enumerate}
Therefore, $q_1[1]=\sigma_j$, and $p_j$ has a prefix $q_0\sigma_j$, and since a periodic factor must occur at least twice, $|q_0|<|p_j|$, contradicting the minimality of $p_j$. If $q_1$ is empty, it follows that $q_0[1]=a_j$ and $p_j$ can be replaced with the empty string $\varepsilon$. 

The string $p_j$ has $p_{j-1}$ as a non-overlapping prefix and suffix. Therefore, since $p_1=\varepsilon$, the string $p_{i-1}$ matches the pattern $\ZP{i-2}$. \qed

\end{proof}

Next, we need to prove that the lexicographically minimal string has the globally minimal alphabet size. To do so, let us cite an additional lemma:

\begin{lemma}[Proved in~\cite{nagashita_et_al:LIPIcs.CPM.2023.23}]\label{l:tomo-pal}
  Let $w$ and $u$ be strings that have the same Manacher array, and let $i$ and $j$ be integers satisfying $1 \le i < j \le |w|=|u|=n+1$. If $w[i + 1..n]$ and $w[j + 1..n]$ are palindromes and $w[i] = w[j]$, then $u[i + 1..n]$ and $u[j + 1..n]$ are palindromes and $u[i] = u[j]$.
\end{lemma}

We can now conclude the proof.

\begin{proof}
    Let us assume that for an array $\man$ and a lexicographically minimal string $S$, there exists a string $S'$ with a strictly smaller alphabet. 

    Looking at the algorithm by~\cite{tomohiro:10} and the proof of~\cref{l:expo}, a new character is introduced to the reconstruction only if no existing character can satisfy the left-to-right reconstruction of the array $\man$. Therefore,  if $|{\Sigma}_S| > |{\Sigma}_{S'}|$ then there exists a minimal index $j$ such that $\sigma_i=S[j]$, and $\sigma_j\notin \Sigma_{S'}$. In~\cref{l:expo}, we proved that the new character $\sigma_i$ is preceded by $(i-1)$ maximal palindromic suffixes, each of them preceded by a different alphabet symbol. However, since a new character was not introduced at $S'[j]$, at least two palindromic suffixes were preceded by the same character, contradicting~\cref{l:tomo-pal}. \qed
\end{proof}

With that, \Cref{thm:minimal-reconstruction} is proved. \qed

The last lemma for this section combines~\cref{thm:minimal-reconstruction} and~\cref{thm:graph-reconstruction} to solve an open problem proposed by~\cite{tomohiro:10}.
\begin{theorem}[Open problem (3) in~\cite{tomohiro:10}]\label{lem:open}
    Given a set of maximal palindromes $\fp$ and a predefined parameter $k$, we can find a reconstruction of $\fp$ that contains exactly $k$ characters, if possible.
\end{theorem}

The idea is to create an optimal coloring using the algorithm from~\cite{tomohiro:10}, update the coloring to include exactly $k$ colors, and conclude by making a string from the updated coloring. 

\begin{proof}\label{ap:pr:open}
    We begin by applying the reconstruction algorithm by~\cite{tomohiro:10}, to achieve an initial string $T$ with the given fingerprint $\fp$. Due to~\cref{thm:minimal-reconstruction}, the alphabet size of $T$ is the minimum possible. Using~\cref{thm:graph-reconstruction}, we construct the restriction graph $G$ of $\fp$.

    If the number $k$ is smaller than the alphabet size of $T$, or $k$ is greater than the number of vertices in the restriction graph, such a reconstruction is impossible.

    However, if $k$ is in range, color $G$ with $T$ to achieve an optimal coloring $\psi$. Find two vertices $v,u\in G$ that satisfy $\psi(v)=\psi(u)$, and change $\psi$ such that $\psi(u)$ is a new alphabet symbol, and repeat this process $k-|\Sigma_T|$ times. From the resulting coloring $\psi$, construct a string $T'$. The string $T'$ has the fingerprint $\fp$, and has exactly $k$ characters.
    \qed
\end{proof}

\section{Future Work}
Several natural questions remain open. Most notably, the exact number of valid Manacher arrays of length~\(n\) is unknown, and a direct combinatorial characterization of these arrays would significantly deepen our understanding. It is also unclear whether efficient uniform sampling of such arrays is possible. Finally, extending the reconstruction and combinatorial framework to approximate palindromes or palindromes containing wildcards remains a largely unexplored avenue.



\begin{credits}
\subsubsection{\ackname} 
I thank Professor Ron Adin for his guidance and valuable assistance with this problem.

This study was partially funded
by ISF grant No. 168/23.
\end{credits}

\newpage

\bibliography{paper}

\newpage
\appendix


\section{Omitted details}

\begin{definition}[Positive and negative dependencies]
    Let $\mathsf A$ be a set describing dependencies in the string.

    We refer to a dependency of the form:
    \[
    S[i..i+\ell-1]=S[j..j+\ell-1]
    \]
    as a {\em positive} dependency, and we refer to a dependency of the form:
    \[
    S[i..i+\ell-1]\neq S[j..j+\ell-1]
    \]
    as a {\em negative} dependency.

    The length of a dependency---positive or negative---is the length of the required match. In the examples above, the length is $\ell$.
\end{definition}

\begin{claim}\label{c:colnp}
    Reconstructing a string with minimal alphabet from a set of dependencies that contain only length-1 negative dependencies is NP-hard.
\end{claim}
\begin{proof}
    We reduce from graph vertex coloring: Let $G=(V,E)$ be an undirected graph, $n=|V|$. We initialize an empty set of dependencies $\mathsf A$. \\
    For every $(v_i,v_j)\in E$, we add the negative dependency $S[i] \neq S[j]$ into $\mathsf A$. We add no positive dependencies.

    Any string $S$ that satisfies all the dependencies in $\mathsf A$ is a valid vertex coloring for $G$, with the coloring function $\psi:V\to \Sigma$ defined as:
    \[
    \psi(v_i)=S[i]
    \]
    Two neighboring vertices $(v_i,v_j)$ never have the same color, as we required $S[i]\neq S[j]$. Conversely, any proper coloring $\psi:V \to \Sigma$ yields a string $S$ defined by $S[i] = \psi(v_i)$ that satisfies all constraints. Therefore, the minimum alphabet size of any satisfying string equals the chromatic number $\chi(G)$.

    Since computing (or deciding whether $\chi(G) \le k$) is NP-hard, minimizing the alphabet size subject to our constraints is NP-hard. 
\end{proof}

As an immediate corollary, reconstructing a string from an arbitrary (mixed positive/negative) dependency set using a minimal alphabet is also NP-hard.

\end{document}